\documentclass[journal,10pt]{IEEEtran}
\usepackage{blindtext}
\usepackage{graphicx}
\usepackage{amsmath}
\usepackage{graphicx}
\usepackage{amsthm}
\usepackage{multirow}
\usepackage{amsmath}

\usepackage{yhmath}
\usepackage{amssymb}
\usepackage{array}
\usepackage{leftidx}
\usepackage{longtable}
\usepackage{stfloats}
\usepackage{cite}
\usepackage{url}
\usepackage{subfigure}
\usepackage{comment}
\usepackage{color}
 \usepackage{algorithm,algpseudocode,float}
\usepackage{lipsum}



\begin{document}
\theoremstyle{plain}
\newtheorem{thm}{Theorem}
\newtheorem{remark}{Remark}
\newtheorem{lemma}{Lemma}
\newtheorem{prop}[thm]{Proposition}
\newtheorem*{cor}{Corollary}
\theoremstyle{definition}
\newtheorem{defn}{Definition}
\newtheorem{condi}{Condition}
\newtheorem{assump}{Assumption}

\title{Hierarchical Game-Based Multi-Agent Decision-Making for Autonomous Vehicles }
\author{Mushuang Liu, Yan Wan, Frank Lewis, Subramanya Nageshrao, H. Eric Tseng and Dimitar Filev

\thanks{M. Liu and Y. Wan are with the Department of  Electrical Engineering, University of Texas at Arlington, Arlington, TX, USA.
Email: mushuang.liu@mavs.uta.edu, yan.wan@uta.edu.
}
\thanks{F. L. Lewis is with UTA Research Institute, University of Texas at Arlington, Fort Worth, TX. Email: lewis@uta.edu.}

\thanks{S. Nageshrao is with Ford Greenfield Labs, 3251 Hillview Ave, Palo Alto, CA 94304, USA (e-mail: snageshr@ford.com).}
\thanks{H. E. Tseng and D. Filev are with Ford Research and Innovation Center, 2101 Village Road, Dearborn, MI 48124, USA (e-mail: htseng@ford.com,  dfilev@ford.com).}
\thanks{This work is supported by Ford Contract Fast Autonomous Driving Decision based on Learning and Rule-based Cognitive Information and ARO Grant W911NF-20-1-0132.}
}
\maketitle
\begin{abstract}
     This paper develops a game-theoretic decision-making framework for autonomous driving in multi-agent scenarios. 
A novel hierarchical game -based decision framework is developed for the ego vehicle. This  framework  features an interaction graph, which  characterizes the interaction relationships between the ego and its surrounding traffic agents (including AVs, human driven vehicles, pedestrians, and bicycles, and others), and enables the ego to smartly select a limited number of agents as its game players. Compared to the standard multi-player games, where all surrounding agents are considered as  game players, the hierarchical game significantly reduces the computational complexity. In addition, compared to pairwise games,  the most popular approach in the literature, the hierarchical game promises more efficient decisions for the ego (in terms of less  unnecessary  waiting and yielding). 
To further reduce the computational cost, we then propose an improved hierarchical game, which  decomposes the 
hierarchical game
 into a set of sub-games. Decision safety and efficiency are analyzed in both hierarchical games. Comprehensive simulation studies are conducted to verify the effectiveness of the proposed frameworks, with an intersection-crossing scenario as a case study.
\end{abstract}
\section{Introduction}
Autonomous vehicles (AVs) are expected to produce significant societal
benefits, such as   reduced
accidents, reduced energy consumption,  enhanced ride comfort, improved travel efficiency, and improved equity for children, elderly and disables \cite{benefit_1,benefit_2}. One of the biggest technical challenges in autonomous driving is the design of decision-making algorithms to generate  safe and efficient decisions for AVs in diverse  traffic scenarios \cite{decision1,decision2,decision3}. Such designs 
for the ego vehicle must take into consideration the interactions between the ego  and  other traffic agents, including AVs, human driven vehicles, pedestrians, and bicycles, and others \cite{interaction_1,interaction_2,interaction_3}. These interactions can be heterogeneous (i.e., agents affect the ego in different ways), time-varying (i.e., the interactions change with the motion of agents), and highly correlated  (i.e., the interaction between one agent pair may affect the interactions between other agent pairs).   

The game-theoretic approaches have  recently been developed as the mathematical framework to capture the interactions among traffic agents for the decision-making of AVs \cite{game_merge,game_leader,game_sta,game_racing,game_potential,game_1,game_2,Victor,my_ford}. 
 In a game theoretic formulation, each agent aims to optimize its own payoff/interest, which is mutually affected by other agents' actions. Of our interest, paper \cite{game_merge} formulates a two-player non-zero-sum game to model the merging-giveaway interaction between a through vehicle and a merging vehicle.  To capture the vehicles' interactions in lane-changing scenarios, a two-player Stackelberg game and a two-player normal-form game were developed, respectively, in papers \cite{game_sta,Victor}.
  To address the intersection-crossing problems in uncontrolled intersections,  a leader-follower game was designed in paper \cite{game_leader}. To model autonomous racing between two vehicles, paper \cite{game_racing} investigated and compared two non-cooperative games including the Nash equilibrium game and the Stackelberg game.   To deal with the coordination of vehicles on highways, a mixed-integer potential game was developed in paper \cite{game_potential}.  Although the above game-theoretic approaches provide us promising solutions in modeling the interactions among agents, they are computationally intensive with the increase of the number of agents. As such, most of the existing game-theoretic studies focus on the two-agent scenarios only.

To generalize the two-agent games to multi-agent scenarios, pairwise games have been recently developed \cite{Victor,game_sta}. In pairwise games, the ego plays a two-player game with each of its surrounding agents and selects the most conservative decision from the game outcomes as its final decision. Although pairwise game is computationally scalable, it ignores the  interaction correlations between different agent pairs. To be more specific, pairwise games assume that the ego is affected by each of its surrounding agents independently. However, this is usually not true  because the surrounding agents also have interactions among themselves, and they often act as a coherent whole to the ego. In pairwise games,    the "local" interaction between the ego and one surrounding agent is captured, while the "global" interaction among multiple agents is ignored. 

To take into account multi-agent interactions while remaining the tractability of the solution, we propose a novel hierarchical game framework in this paper. This framework features a game player selection step based on an interaction graph that captures  the interactions among agents. With this hierarchical game, the ego is able to select a subset of agents as game players such that the number of players are limited while the important interaction information is also captured. Compared to pairwise games, the proposed hierarchical game is able to generate  more efficient decisions for the ego (in terms of less unnecessary waiting or yielding), while guaranteeing the decision safety. 

This paper is organized as follows.  Section \ref{II} briefly reviews a two-agent game-theoretic framework developed in \cite{my_ford}, which serves as a basis in this paper. Section  \ref{III} introduces the novel hierarchical game that generalizes from the two-agent scenarios to multi-agent scenarios. An improved hierarchical game approach is also proposed in this section to further reduce the computational complexity.  Section \ref{IV} designs game payoffs in the  proposed games. Section \ref{V} shows simulation studies, and Section \ref{VI} concludes the paper.

\section{Preliminaries and Review of Two-Agent Game-Theoretic Decision-Making Framework}\label{II}
In this section, we first briefly review our prior work on the  game-theoretic decision-making framework for two agents \cite{my_ford}. This framework serves as a basis for the development in this paper. 
The normal-form games are also reviewed in this section to facilitate the analysis in this paper.
\subsection{Decision-making framework in two-agent scenarios}
A three-level decision-making framework (shown in Figure \ref{Three_level}) was proposed in \cite{my_ford} to generate  safe and efficient decisions for the ego when it interacts with  one other agent. Here is a brief description of   each level. The first level is an action space filter, which determines the action space of the ego and filters out the actions that violate "hard" traffic rules, e.g., driving through a red traffic light. The second level is a normal-form game based decision-making, where the Nash solution is generated. 
The third level is a final safety check to prevent an unsafe action from being implemented. To better illustrate the role of each level,  let us consider the intersection-crossing scenario as a case study.

\begin{figure}[thpb]
\centering
\includegraphics[width=0.25\textwidth]{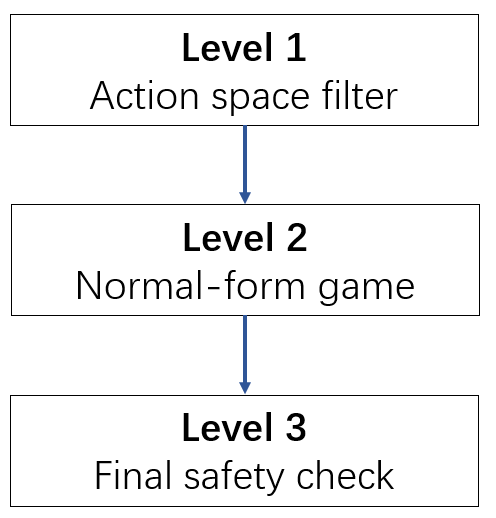}
\caption{Three-level decision-making framework }\label{Three_level}
\end{figure}

Consider a four-way stop sign intersection shown in Figure \ref{scenario_1}. The ego is labeled with the number "1" and aims to go straight. Meanwhile, the vehicle 2 is also approaching the intersection and aims to cross it. The action space for the two vehicles is $A_i=\{\text{yield, go}\}$ for $i=1,2$. To generate appropriate decisions, the ego needs to consider 1) the rules associated with the stop sign, and 2) the interactions with vehicle 2. Using the three-level framework in Figure \ref{Three_level}, the ego determines its action space in the first level by considering the hard rule that all vehicles must come to a complete stop before they cross the stop sign. In the second level, the ego plays a two-player normal-form game  with the vehicle 2 (as shown in Figure \ref{game_intersection}), and selects the Nash decision. The game payoffs $a_{ij}$ and $b_{ij}$ ($i,j\in \{1,2\}$) are expected to capture both the first-come-first-go rule and the safety requirement related to the vehicle 2's states and possible actions. A detailed design procedure can be found in \cite[Section IV]{my_ford}. The third level checks whether the Nash solution is safe. If it is not safe, then a backup plan of safe actions ("yield" in this case) is implemented.  

\begin{figure}[thpb]
\centering
\includegraphics[width=0.3\textwidth]{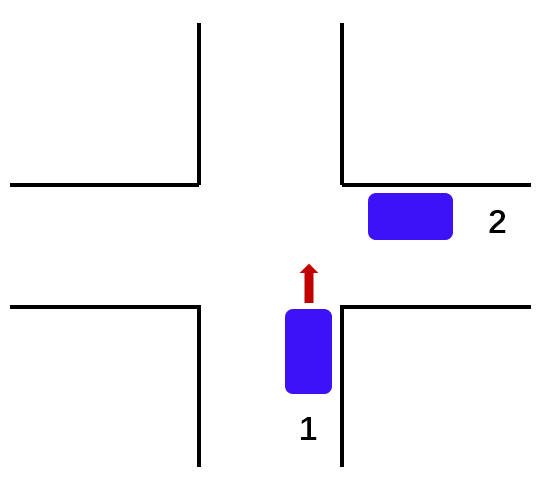}
\caption{The intersection-crossing scenario }\label{scenario_1}
\end{figure}

\begin{figure}[thpb]
\centering
\includegraphics[width=0.31\textwidth]{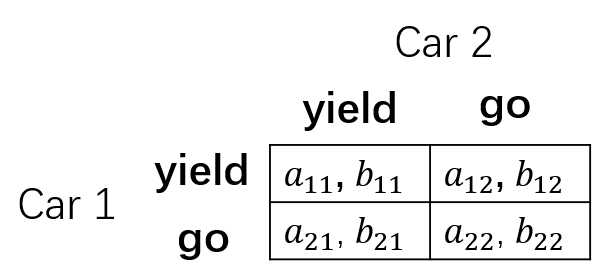}
\caption{Normal-form game of intersection-crossing}\label{game_intersection}
\end{figure}

By performing the above decision-making process periodically for every few milliseconds, the ego is able to make timely decisions subject to  changes in environments (e.g., vehicle 2's change of actions). The effectiveness of this framework has been verified in two-agent scenarios \cite{my_ford}. The focus of this paper is to  generalize this framework to multi-agent scenarios. As the generalization of the first and third levels are straightforward, we focus our study on the  the second level, i.e., the normal-form game -based decision-making.  Before we present our solution, let us briefly review the  multi-player normal-form game in the next subsection. 

\subsection{Review of normal-form games}\label{sec:pre}
 A multi-player normal-form game is defined as a tuple ($\mathcal{N},A,J$), where  $\mathcal{N}=\{1,2,\cdots,N\}$ is the set of $N$ players, $A=A_1\times A_2 \times \cdots \times A_N$ is the set of actions with $A_i$ the set of  player $i$'s actions, and $J=J_1\times J_2 \times \cdots \times J_N$, with $J_i:A\rightarrow \mathbb{R}$ the action-dependent payoff functions of player $i$ \cite{normal_game}. In a normal-form game, each player aims to maximize its payoff by selecting the optimal action that takes into account other players'  payoffs and possible actions.

\begin{figure}[thpb]
\centering
\includegraphics[width=0.35\textwidth]{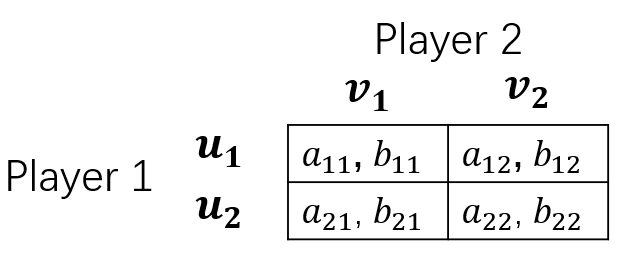}
\caption{Two-player normal-form game }\label{normal}
\end{figure}

A payoff table, e.g., Figure \ref{normal}, is usually employed to represent a normal-form game.  In the two-player game shown in Figure \ref{normal}, each player has two possible actions: $A_1=\{u_1,u_2\}$ and $A_2=\{v_1,v_2\}$. The payoffs for the two players are denoted by $a_{ij}$ for player 1  and $b_{ij}$ for player 2, respectively, where $i,j\in\{1,2\}$ represent the associated actions. To be more specific, the payoff pair $(a_{11}, b_{11})$ represents that if player $1$ selects action $u_1$ and player $2$ selects $v_1$, then the two players receive the payoff $a_{11}$ and $b_{11}$, respectively.   

To consider normal-form games with enlarged action spaces, additional rows and columns should be added in the payoff table. To consider multiple players in the  normal-form game (i.e., $N>2$),  the dimension of the table is required to be extended from $2$ to $N$.

Next let us define two of the most important concepts in normal-form games:  "best response" and "Nash equilibrium". 
\begin{defn} [Best Response]
The best response of player $i$ is defined as the action $s_i^*$ of player $i$ such that 
\begin{equation}
    J_i(s_i^*,s_{-i})\geq J_i(s_i,s_{-i})
\end{equation}
for all possible $ s_i\in A_i$, where $s_{-i}$ is the action set performed by all other agents, i.e.,  $s_{-i}=\{s_1,\cdots,s_{i-1},s_{i+1},\cdots, s_N\}$.  
\end{defn}
\begin{defn} [Nash equilibrium]
An $N$-tuple of actions $\{s_1^*,s_2^*,...,s_N^*\}$ is said to be a Nash equilibrium for an $N$-player game if for all $i=1,2,\cdots,N$,  
\begin{equation}
    J_i(s_i^*,s_{-i}^*)\geq J_i(s_i,s_{-i}^*)
\end{equation}
holds for all possible $ s_i\in A_i$.
That is, if all players perform their best responses against each other, then a Nash equilibrium is achieved.

Nash equilibrium is important as it represents a stable status of a game in the sense that no player is able to get a better payoff by unilaterally changing its decision. In other words, if a Nash equilibrium is achieved in a game, then no player would have the incentive to  change its decision. 
\end{defn}

\section{Hierarchical Game -Based Multi-Vehicle Decision-Making}\label{III}
In this section, we develop a novel approach  to generalize the two-agent decision-making framework described in Section \ref{II} to multi-agent scenarios. 
 We first introduce two existing solutions developed in the literature in Section \ref{IIIA}, from which our novel solution, called hierarchical game, is motivated and developed in Section \ref{IIIB}. Properties of the  hierarchical game are described in Section \ref{IIIC}. An improved hierarchical game  is  presented in Section \ref{IIID} to  further reduce the computational complexity. 

\subsection{Existing solutions}\label{IIIA}
Two   approaches have been widely used  in the literature to solve multi-agent game-theoretic decision-making: 1) multi-player game with all surrounding agents as game players, and 2) pairwise game.  

In a multi-player game where all agents are counted as game players, the global interaction among agents is captured. This formulation, in principle, can lead to  optimal decisions. However, as the computational complexity increases  exponentially with the number of agents, this formulation is not practical in solving real-time decision-making problems in practice \cite{complexity}.  

To provide a computationally affordable solution, pairwise game has been recently developed in the literature  \cite{game_sta,game_leader,Victor}. In pairwise game design, the ego plays a two-player game with each of its surrounding agents and selects the most conservative decision as its final decision \cite{Victor}. This solution is computationally scalable with safety  guarantees. 
However, because only a pair of agents are considered in each game, the correlations among agent pairs are ignored. In other words,  the global  interaction among agents cannot be captured by pairwise games. Therefore, the  solution from pairwise games can be excessively conservative and may lead to the freezing vehicle symptom, i.e., all AVs decide to yield to each other.

To overcome the above challenges, we introduce a novel hierarchical game approach, which  balances the solution optimality and computational complexity to generate safe and efficient decisions with affordable computational cost.

\subsection{Hierarchical game design}\label{IIIB}
 The intuition behind the hierarchical game is as follows. We notice that the surrounding agents are not equally important to the ego  to make its decision, and only a limited number of agents are critical.  To  illustrate this perspective, let us consider the intersection-crossing scenario pictured in Figure \ref{multi_example}. The ego (labeled with the number "1")  aims to go straight and cross the intersection. Meanwhile, agents 2-6 are also going to cross the intersection. To avoid crashing, the ego has to keep monitoring the behaviors of pedestrian 2 and vehicle 3 (colored with green), as their planned trajectories conflict with the ego's future trajectory. On the other hand,  agent 3's behaviors are dependent on the behaviors of agents 4 and 5 (colored with blue), as the future trajectory of agent 3 conflicts with the trajectories of both agents 4 and 5. As such, taking into consider the behaviors of agents 4 and 5 also benefits the ego's decision-making. 
 In contrast to these  agents, the vehicle 6 and other agents colored in yellow are not very relevant to the ego's decision-making, as their decisions hardly affect the ego's decision. This example motivates that characterizing the interactions between the ego and its surrounding agents is important, as it helps the ego to select the most "relevant" agents to focus on.  To characterize such interaction relationships, let us first define the following three concepts, called \textit{trajectory conflicts}, \textit{ego's direct neighbors}, and \textit{ego's $k^{th}$ level of neighbors}.

\begin{figure}[thpb]
\centering
\includegraphics[width=0.28\textwidth]{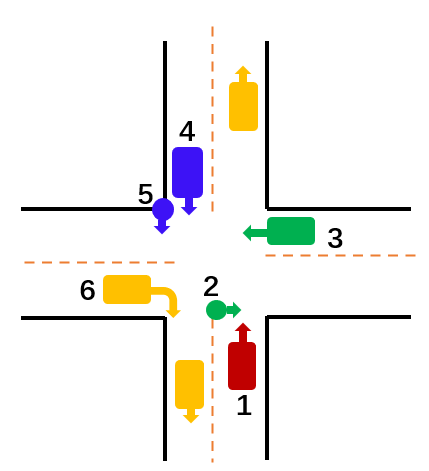}
\caption{An intersection-crossing example with multiple agents }\label{multi_example}
\end{figure}

\begin{defn} [Trajectory conflicts]\label{d3}
Consider an intersection-crossing scenario, agents $i$ and $j$  are said to have trajectory conflicts if their future trajectories  intersect with each other within the intersection region.
\end{defn}
\begin{defn} [Ego's direct neighbors]\label{d4}
The ego's direct neighbors (or called the $1^{{st}}$ level of neighbors)  are defined as the agents that have trajectory conflicts with the ego.
\end{defn}
\begin{defn} [Ego's $k^{{th}}$ level of neighbors]\label{d5}
The ego's $k^{{th}}$ level of neighbor are defined as the agents that have trajectory conflicts with the ego’s $(k-1)^{{th}}$ level of neighbors.
\end{defn}

\begin{figure}[thpb]
\centering
\includegraphics[width=0.22\textwidth]{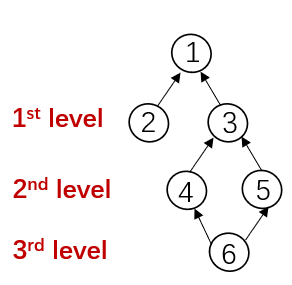}
\caption{Interaction graph for agents in Figure \ref{multi_example} }\label{interaction_graph}
\end{figure}
With Definitions \ref{d3}-\ref{d5}, we can compose an interaction graph to describe the agent interactions for the scenario in Figure \ref{multi_example}. The interaction graph is shown in Figure \ref{interaction_graph}, where the nodes represent the agents and the arrows represent the direct dependence: the ego's decision is directly decided by agents 2 and 3, and so on. For this specific example, the ego's surrounding agents can be grouped into three levels. The first level includes agents 2 and 3, which have trajectory conflicts with the ego. The second level includes agents 4 and 5, both of which directly affect the agent 3's decisions. The third level is the agent 6, which determines the decisions of  agents 4 and 5. This interaction graph delineates the importance of each agent to the ego. The agents in lower levels (e.g., agents 2 and 3) have stronger interactions with the ego compared to the ones in higher levels (e.g., agent 6). As such, considering the constraint of limited computational resource, the agents in lower levels should have higher priority to be included in the game player set.  In other words, if it is not possible to include every agent as the game players, then the agents in higher levels should be dropped, as they present weak interactions with the ego and are less significant in affecting the ego's decision. 

The general interaction graph for  an intersection-crossing scenarios is pictured  in Figure \ref{general}, from which we develop our hierarchical game framework. In a hierarchical game, the ego is designed to play a multi-player game with a subset of its surrounding agents, called the first $k^{\text{th}}$ level of neighbors. The selection of the number $k$ depends on the ego's computational capability. To be more specific, if the ego is expected to generate a decision  every $\Delta t$, and its processor allows at most a $N_{max}$-player game within $\Delta t$ computational time, then the number of $k$ should be selected as the maximum value that satisfies the following condition:
The total number of agents in the first $k^{\text{th}}$  levels are less than or equal to $N_{max}$. 
For example, consider the interaction graph in Figure \ref{interaction_graph}. If we expect the ego to repeat the decision-making process periodically  every 10 milliseconds, during which the processor is able to solve a $5$-player game at most (i.e., $N_{max}=5$), then $k$ should be selected as $2$ because a total of $5$ agents are included in the first two levels. On the other hand, if $N_{max}>6$, then we can select $k=3$ to derive a global Nash solution where all agents are included in this case.     

\begin{figure}[thpb]
\centering
\includegraphics[width=0.35\textwidth]{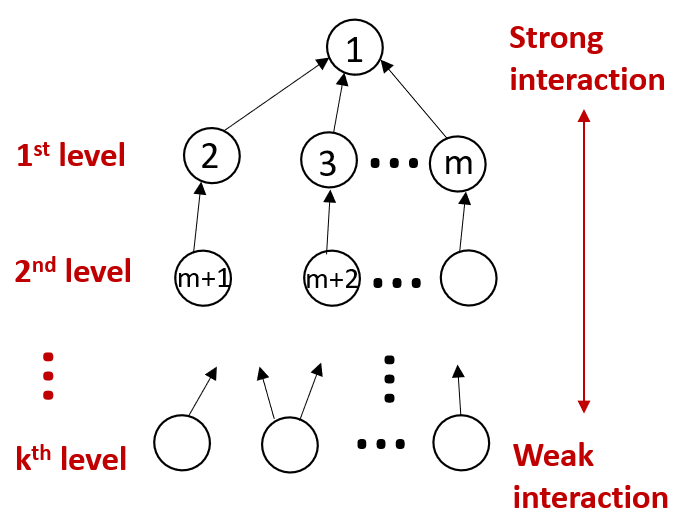}
\caption{General interaction graph }\label{general}
\end{figure}

In addition to characterizing the interaction relationships, we notice that the  clustering of agents is also important in determining the interaction graph to reduce the number of game players. 
Consider the  intersection-crossing scenario shown in Figure \ref{clustering_example}. Let the vehicle labeled with the number "1" be the ego. Clearly, all pedestrians labeled with green circles have trajectory conflicts with the ego, and as a result, all of them should be included in the ego's first level neighbors. However, it is not be necessary to include all of them in the ego's game player set, as their interactions with the ego are very similar (i.e., they all aim to cross the sidewalk in front of the ego). 
Motivated by this,  we propose a simple and effective agent clustering criterion in Definition \ref{d6} to group the agents that have similar interactions with the ego. It is expected that one clustered agent group can be regarded as one game player, and if the ego is safe to this game player, then it is safe to all agents in this group. 

\begin{figure}[thpb]
\centering
\includegraphics[width=0.42\textwidth]{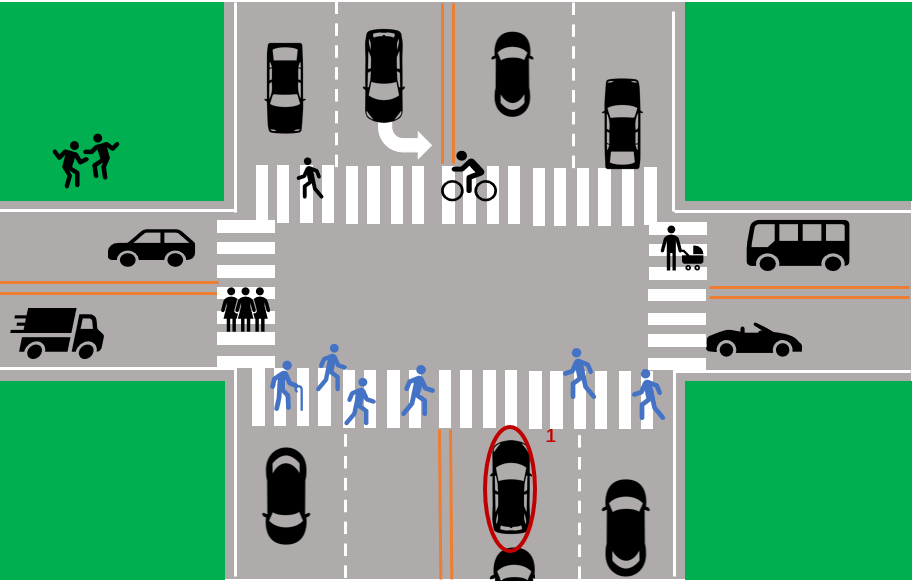}
\caption{Illustrative example to motivate agent clustering }\label{clustering_example}
\end{figure}

\begin{defn} [Agent clustering]\label{d6}
We cluster the agents with parallel trajectories within the intersection region into one group. 
\end{defn}

This agent clustering criterion is simple to implement as one only needs to check the parallelism of two trajectories. Although simple, this approach is effective in the clustering and reducing the number of game players. For example, consider the scenario shown in Figure \ref{clustering}. With the criterion defined in Definition \ref{d6}, all agents colored in blue are clustered into one group, which dramatically reduces the number of the first level neighbors for the ego. It is also reasonable because it meets our expectation that  the ego is safe to all of them as long as it is safe to the "most conflicting" one. 

Note that definition \ref{d6} provides one possible approach to  cluster agents. 
Other agent clustering methods designed based on  trajectory similarity measures (e.g., \cite{clustering_1,clustering_2}) can also be incorporated into our hierarchical game framework. 
\begin{figure}[thpb]
\centering
\includegraphics[width=0.42\textwidth]{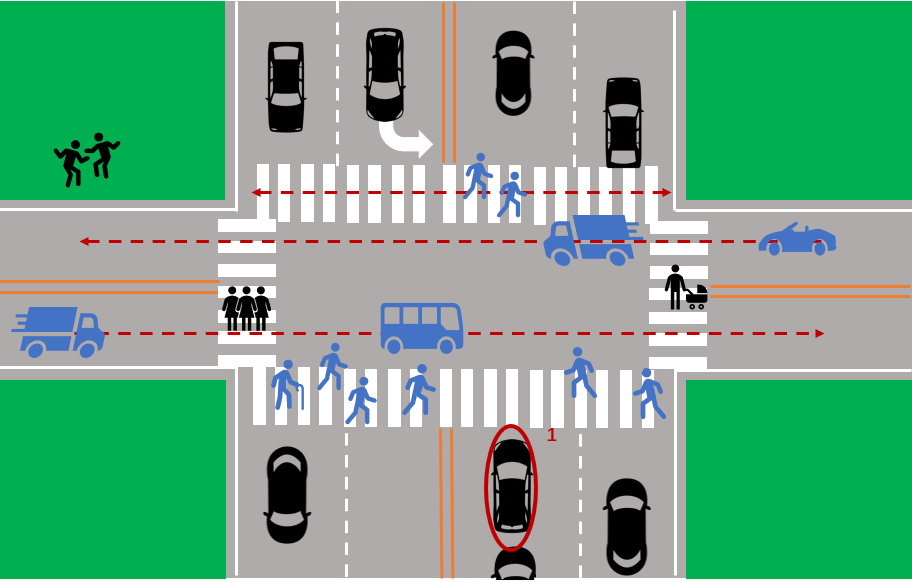}
\caption{Illustrative example to show the effectiveness of agent clustering }\label{clustering}
\end{figure}

Next, let us consider the selection of the "most conflicting" agent from a clustered group, such that the clustered agent group can be represented by the selected agent. 
To measure the "conflicting level", let us first define time-to-collision (TTC).  
\begin{defn} [Time-to-collision]\label{d7}
The time-to-collision  $T^c_{ij}$ is defined as the time that agent $i$ takes to reach the relative position of the agent $j$.
\end{defn}
Consider the vehicles $1$ and $2$ in Figure \ref{TTC_fig}. $T^c_{12}$ (and $T^c_{21}$, respectively) is defined as the time that the vehicle 1 (vehicle 2) takes to reach the relative position of vehicle 2 (vehicle 1), i.e.,  
\begin{equation}\label{eq:ttc}
    T^c_{12}=\frac{d_{12}}{v_1},
\end{equation}
\begin{equation}
    T^c_{21}=\frac{d_{21}}{v_2},
\end{equation}
where $d_{12}$ and $d_{21}$ are, respectively, the longitudinal and latitudinal distances between vehicles 1 and 2 along  vehicle 1's heading and vehicle 2' heading, and  $v_1$ and $v_2$ are the speed of vehicle 1 and 2, respectively.
\begin{figure}[thpb]
\centering
\includegraphics[width=0.3\textwidth]{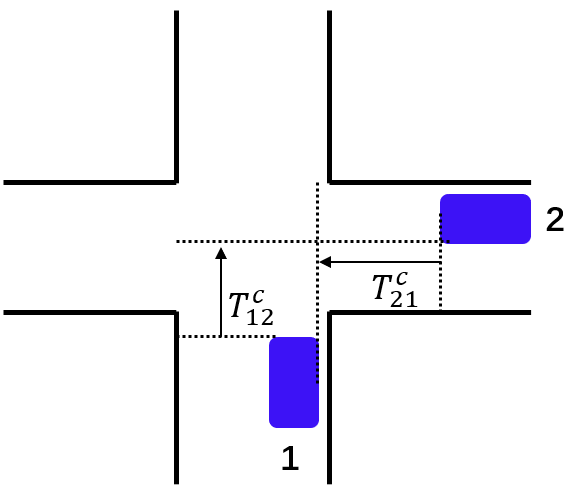}
\caption{Time-to-collisions}\label{TTC_fig}
\end{figure}

With the defined TTC, we propose the following representative agent selection criterion. 

\begin{defn} [Representative agent selection]\label{d8}
An agent $j$ in the agent $i$'s first level of neighbors is selected as a representative agent of a clustered group if $j=\text{argmin}_j|T^c_{ij}-T^c_{ji}|$ holds for all $j$ in this group.
\end{defn}

Definition \ref{d8} uses the difference between two TTCs, i.e., $|T^c_{ij}-T^c_{ji}|$, to measure the conflicting level between the agents $i$ and $j$. $|T^c_{ij}-T^c_{ji}|=0$ represents that the two agents reach the same location at the same time, i.e., a crash happens. On the other hand,  $|T^c_{ij}-T^c_{ji}|$ being large means that the two agents are not likely to crash into each other. As such, the agent $j$ that leads to the minimum $|T^c_{ij}-T^c_{ji}|$ is  the "most conflicting" one to the agent $i$. 

Note that both agent clustering and representative agent selection are procedures in determining the interaction graph. After deriving the interaction graph, the ego can then perform  the hierarchical game.  The detailed steps of performing a  hierarchical game are described in Algorithm \ref{a1}. 

\begin{algorithm}[t]
\caption{Hierarchical Game  in Intersection-Crossing} \label{a1}
\hspace*{0.0in} {\bf Input:} \\ 
\hspace*{0.02in}Total number of agents: $N$;\\ 
\hspace*{0.02in}Desired time period for each play: $\Delta t$;\\ 
\hspace*{0.02in}Maximum number of game players that  can be afforded for each $\Delta t$: $N_{max}$;\\ 
\hspace*{0.02in}Positions, velocities, and future trajectories of all agents in the intersection.  \\
\hspace*{0.02in} {\bf Output:} \\
\hspace*{0.02in} The ego's decision: \textit{Go} or \textit{Yield}.\\
\hspace*{0.02in} {\bf Procedures:} 
\begin{algorithmic}[1]
\State {\bf For} $t=\Delta T,2 \Delta T,\cdots$ {\bf do}
\State \hspace*{0.2in} Determine the interaction graph according to steps 3-8.
\State \hspace*{0.2in}  Cluster agents with parallel trajectories.
\State \hspace*{0.2in} Set $j=1$ and $\mathcal{N}_0=\{1\}$.
\State \hspace*{0.2in}{\bf While} $\mathcal{N}_{j-1}\neq \emptyset$ {\bf do}
\State \hspace*{0.2in} Determine the ego's $j^{th}$ level of neighbors (denoted by $\mathcal{N}_{j}$) according to the agents' trajectories and representative agent selection. If no agent belongs to $\mathcal{N}_{j}$, then set $\mathcal{N}_{j}= \emptyset$.
\State \hspace*{0.2in} Set $j=j+1$.
\State \hspace*{0.2in}{\bf End while}
\State \hspace*{0.2in} Select $k$ according to the derived interaction graph, $\Delta t$, and $N_{max}$. 
\State \hspace*{0.2in}Play a multi-player normal game with the first $k^{th}$ level of agents and select the Nash solution as the ego's decision.
\State {\bf End for}
\end{algorithmic}
\end{algorithm}

\subsection{Properties of hierarchical game} \label{IIIC}
The developed hierarchical game provides a trade-off between the solution optimality and computational complexity. The properties of the derived solution, in terms of safety and efficiency,  are described respectively in the following two propositions.

\begin{prop}
 Consider a hierarchical game from Algorithm \ref{a1}. If  no agent clustering is implemented and all levels of neighbors are included as game players, then the hierarchical game generates the global Nash solution.
\end{prop}
\begin{proof}
If no agent clustering is implemented and all levels of neighbors are included, then the  hierarchical game is equivalent to an $N$-player game where all agents are included in the game player set. As all agents are included, the global interactions are captured, and as such, the global Nash is derived. 
\end{proof}
\begin{prop}
Consider a hierarchical game from Algorithm \ref{a1}. Assume that the game payoffs are well-designed such that the Nash decision is a safe decision  for the game players considering the interactions captured by the game. Then even though only a subset of surrounding agents are included in the game player set, the ego's decision is still guaranteed to be safe. 
\end{prop}
\begin{proof}
According to the hierarchical game, the ego includes its first $k^{\text{th}}$ ($k\geq1$)  level of neighbors as its game players. As the first level neighbors are the only agents (among all surrounding agents) that have trajectory conflicts with the ego, the ego is safe as long as the interactions with the first level neighbors are captured. As $k\geq1$ always holds,  the safety of the ego is always guaranteed.  
\end{proof}
 We can conclude from the above properties that, in general, the more levels of neighbors are included in the game, the closer to global optimal  of the derived solution is (i.e., the more efficient of the decision is).   However, the better solution  is derived at the cost of more computational resources. 

In addition, it is  worthy  commenting  that  the hierarchical game design is consistent with  humans' decision-making process. When we drive to cross intersections with heavy traffic, we do not pay equal attention to every agent around us. Instead, we usually keep monitoring  the agents who are the "most conflicting" to us (i.e., the ones that may lead to  accidents with us). In addition, if we can handle, we would also pay attention to the "relevant" agents,  who directly affect the behaviors  of the "most conflicting" agents. In contrast, we hardly pay attention to the agents that have nothing to do with our planned trajectories.   As such, from this point of view, the hierarchical game  naturally mimics the human-driving decision-making process. 


 Next, let us consider a comparison between the hierarchical game and the two existing approaches: $N$-player game and pairwise game. Compared to the $N$-player game, the hierarchical game is able to reduce the computational complexity from $2^N$ to $2^{N_h}$, where $N_h$ is the number of game players in the hierarchical game. In complex scenarios, $N_h$ can be much less than $N$. For example, consider the intersection-crossing scenario in Figure \ref{comparison_multi}. If the ego considers all surrounding agents as game players, a $26$-player game needs to be solved, which can be impossible in a limited time duration. On the other hand, a hierarchical game with the first two levels of neighbors leads to a $4$-player game (the interaction graph is shown in Figure \ref{interaction_comparison}), and the generated decisions can be both safe and efficient.   
 
 Compared to the pairwise game, the hierarchical game provides more efficient decisions, as more agent interactions are captured. To illustrate this point, let us consider the  scenario pictured in Figure \ref{pairwise_example}, where two vehicles and one pedestrian are involved.  The future trajectories of the three agents are plotted in dashed lines. Consider the vehicle 1 as the ego. According to the hierarchical game, the vehicle 2 and pedestrian 3 are the ego's first and second level neighbors, respectively.  By playing a 3-player game with the two agents, the ego should be motivated to go at this moment because its first level agent, i.e., vehicle 2, is waiting for pedestrian 3 and is not going to move for the next few seconds. On the other hand, if the ego plays pairwise games, then the decision of "yield" would be selected.  It is because that the game with the vehicle  2 generates the ego a "yield" decision, as agent 2 comes to the intersection earlier and has higher road priority. As we can see from this example,  pairwise games may lead to  inefficient decisions for the ego, because they only capture local interactions and neglect the correlations between interactions. 

\begin{figure}[thpb]
\centering
\subfigure[]{\label{comparison_multi}
\includegraphics[width=0.42\textwidth]{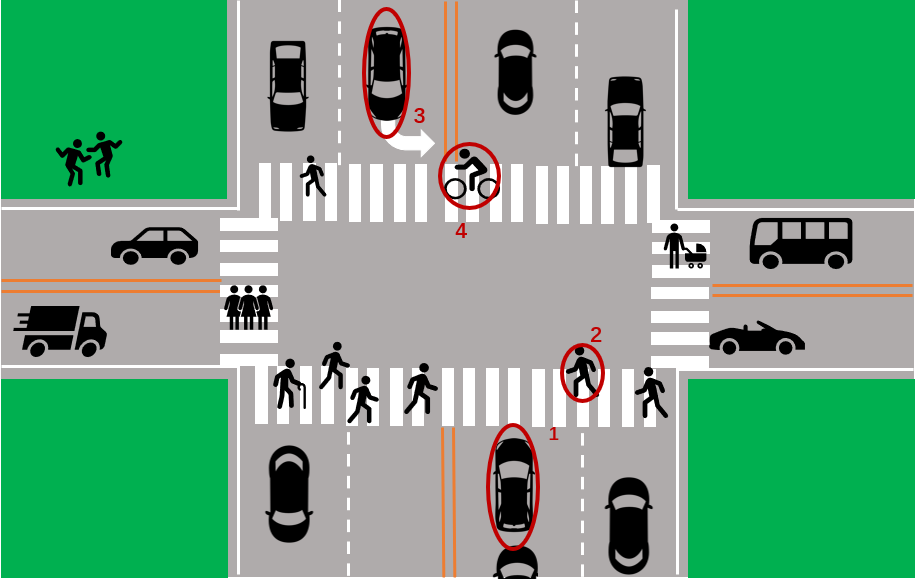}}
\subfigure[]{\label{interaction_comparison}
\includegraphics[width=0.24\textwidth]{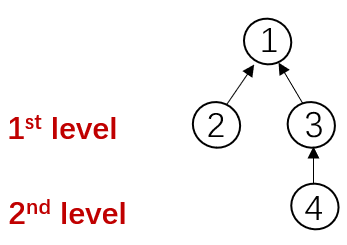}}
\caption{(a) Illustrative example for the comparison between hierarchical game and $N$-player game, and (b) Corresponding interaction graph}
\end{figure}

 \begin{figure}[thpb]
\centering
\includegraphics[width=0.25\textwidth]{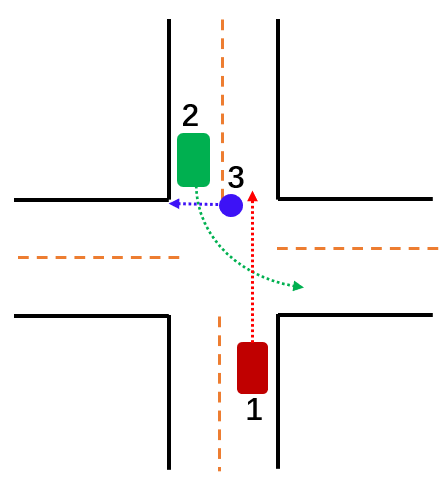}
\caption{Illustrative example for the comparison between hierarchical game and pairwise games.}\label{pairwise_example}
\end{figure}


\subsection{Improved hierarchical game design} \label{IIID}
In this subsection, we propose an improved hierarchical game to further reduce the computational complexity of a $N_h$-player hierarchical game.  In the improved hierarchical game, we decompose the  $N_h$-player game into a set of $N_j$-player games, where $j=1,2,\cdots,M$, $M$ is the number of sub-games, and $N_j$ is the number of game players in the $j^{th}$ sub-game. 

To illustrate the design of the improved hierarchical game, let us first define a new concept called \textit{branch} in an interaction graph.  The decomposition of the  $N_h$-player game is performed by checking whether or not two branches share common nodes. 
\begin{defn} [Branch]\label{branch}
A branch in an interaction graph is defined as an interaction tree where a first level agent serves as the root node and all agents that have a directed path to the root are included in the branch. The number of branches in an interaction graph equals the number of first level agents.
\end{defn}

 The intuition behind the game decomposition is described as follows.  Given the fact that the ego is safe if and only if it is safe to all agents in the first level, it is reasonable to consider the first level agents separately in different sub-games, if they do not have interactions with each other. To be more specific, let us consider  
 the interaction graph in Figure \ref{interactoin_threesub}, where three branches are involved. Branch 1 is  $\{2,5\}$, Brunch 2 is $\{3,6\}$, and Branch 3 is $\{4\}$. In addition, these three branches do not have interactions  except  that they all affect the ego's decision. As such, it is reasonable for the ego to consider these three branches separately in three sub-games, and then select the most conservative decision to ensure that it is safe to all of its first level neighbors.   In other words, instead of playing a  $6$-player, in the improved hierarchical game, the ego is designed to play three sub-games: two $3$-player games with players $\{1,2,5\}$ and $\{1,3,6\}$, respectively, and a $2$-player game with players $\{1,4\}$. In contrast, if two branches have interactions with each other, then the two branches cannot be decomposed. Consider the interaction graph in Figure \ref{interaction_twosub}, where the last two branches have interactions: they are both affected by agent 6. In this case, the four agents in the two branches, i.e., agents 1, 3, 4, and 6, need to be included in one game together to capture the global interactions.  As a result, the $6$-player game in this case is decomposed into two sub-games: a $3$-player game with the player set $\{1,2,5\}$ and a $4$-player game with the player set $\{1,3,4,6\}$.  The general rules for the decomposition of a $N_h$-player game  is given in Definition \ref{decomposition}.

\begin{defn} [Interaction graph decomposition]\label{decomposition}
If  two branches in an interaction graph do not share common agents,  then they can be considered  in two separate sub-games.  
\end{defn}
\begin{figure}[thpb]
\centering
\subfigure[]{\label{interactoin_threesub}
\includegraphics[width=0.22\textwidth]{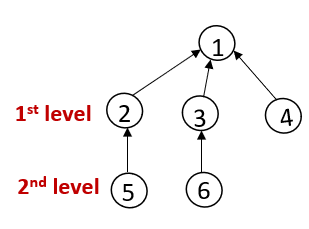}}
\subfigure[]{\label{interaction_twosub}
\includegraphics[width=0.24\textwidth]{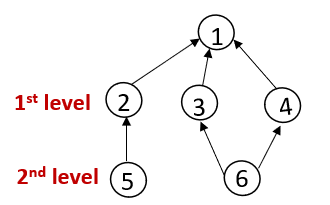}}
\caption{Examples of interaction graphs with (a) three branches and (b) two branches}\label{graph}
\end{figure}

With the above improved hierarchical game design, the computational complexity is further reduced from $2^{N_h}$ to $\sum_{j=1}^{M}2^{N_j}$ with guaranteed safety, where $N_j\leq N_h$. 

\section{Payoff Design}\label{IV}
 After the selection of game players,  the ego needs to determine its payoff matrix to perform a game. The payoffs are desired to capture both the safety requirement and soft traffic rules, as illustrated in \cite{my_ford}. 
 In this section, we describe the design of payoff matrix in an $N_j$-player normal-form game, including the safety-based payoff in Section \ref{IVA} and the  rule-based payoff in Section \ref{IVB}. 
 
 Denote the action-dependent payoff function for agent $i$ as $J_i(a_1,a_2,\cdots,a_{N_j})$, where $i=1,2,\cdots,{N_j}$ and $a_i$ is the  action of agent $i$.  Consider a stop-sign intersection-crossing scenario, then the action space for each agent is $A_{i}=\{go, yield\}$ for all $i=1,2,\cdots,{N_j}$. The payoffs $J_i(a_1,a_2,\cdots,a_{N_j})$ are designed as   
\begin{equation}\label{aij}
\begin{split}
    J_i(a_1,a_2,\cdots,a_{N_j})=&\beta J_i^s(a_1,a_2,\cdots,a_{N_j})\\
    &+(1-\beta )J_i^r(a_1,a_2,\cdots,a_{N_j}),
    \end{split}
\end{equation}
where $J_i^s(a_1,a_2,\cdots,a_{N_j})$ is  the safety-based payoff to measure how safe an action vector is for agent $i$, $J_i^r(a_1,a_2,\cdots,a_{N_j})$ is the rule-based payoff to capture the reward gained by obeying the traffic rule, and $\beta$ is a weighting parameter to tune the ego's behavior and satisfies $0\leq\beta \leq1$. 
Next let us  design  the two  payoffs respectively.
\subsection{Safety-based payoff}\label{IVA}
 We use time-to-collision defined in Definition 7 and time-of-safe-crossing (TOSC) defined in Definition \ref{TOSC} to design the safety-based payoff. 

\begin{defn} [Time-of-safe-crossing]\label{TOSC}
The time-of-safe-crossing  $T^s_{i}$ is defined as the time that agent $i$ takes to completely cross the intersection.
\end{defn}
Consider the vehicles $1$ and $2$ in Figure \ref{TTC}. $T^s_{1}$ ($T^s_{2}$) is defined as the time that the vehicle 1 (vehicle 2) takes to completely cross the intersection. 
\begin{equation}
    T^s_{1}=\frac{d_{1}}{v_1},
\end{equation}
\begin{equation}
    T^s_{2}=\frac{d_{2}}{v_2},
\end{equation}
where  $d_{1}$ is the longitudinal distance between the vehicle 1 and the end of the intersection along the direction of vehicle 1's heading, and $d_{2}$ is the latitudinal distance  between  vehicle $2$ and the end of the intersection along the vehicle 2's heading.

\begin{figure}[thpb]
\centering
\includegraphics[width=0.3\textwidth]{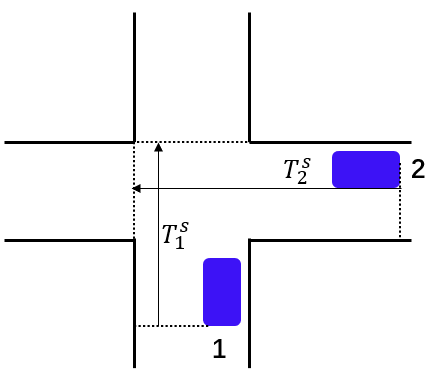}
\caption{Definition of time-of-safe-crossing}\label{TTC}
\end{figure}

To design the safety-based payoff for the ego (i.e., $J_1^s(a_1,a_2,\cdots,a_{N_j})$), let us  consider a two-level $N_j$-player game. We assume, without loss of generality, that agents $\{2,3,\cdots,m\}$ are in the first level and agents $\{m+1,m+2,\cdots,N_j\}$ are in the second level. As only the first level agents have trajectory conflicts with the ego, we design the ego's safety-based payoff as follows. 
\begin{equation}\label{safety_1}
    J_1^s(Y,a_2,...,a_N)=\theta_1(T_{1}^s-\theta_2T_{21}^c)+\cdots+\theta_1(T_{1}^s-\theta_2T_{m1}^c),
\end{equation}
\begin{equation}\label{safety_2}
    J_1^s(G,a_2,...a_N)=\theta_3(T_{21}^c-\theta_4T_{1}^s)+\cdots+\theta_3(T_{m1}^c-\theta_4T_{1}^s),
\end{equation}
where "Y" and "G" represent the actions of  "yield" and "go", respectively, and $\theta_1,\theta_2,\theta_3,$ and $\theta_4$ are constant parameters to tune the ego's behaviors. 
Equation \eqref{safety_1} represents that the ego is motivated to yield as long as there exists at least one agent in the first level such that $T^s_1$ is much larger than $T^c_{21}$, i.e., the ego is not safe to cross the intersection. Equation \eqref{safety_2} represents that the ego is motivated to go if for all $k\in\{2,\cdots,m\}$, $T_{k1}^c$ is much larger than $T^s_1$, i.e., the ego is safe to all first level neighbors.

In addition, for an agent $k$ in the first level, i.e., $k\in\{2,\cdots,m\}$, if at least one of the agent $k$'s direct neighbors  select to go, then the term $\theta_3(T_{k1}^c-\theta_4T_{1}^s)$ should be replaced by $\theta_3(T_{k1}^c-\theta_4T_{1}^s+R)$, where $R$ is a reward with a positive value. This reward is to capture the interaction between the ego and its second level neighbors. If agent $k$'s direct neighbor(s) select to go, then agent $k$ is more likely to yield, and as a result, the safety-check condition for the ego, i.e., $\theta_3(T_{k1}^c-\theta_4T_{1}^s)$, should be easier to satisfy, which is realized by the addition of a positive reward, $R$.

The above payoff design is developed for the ego, i.e., $J_1^s(a_1,a_2,...a_N)$. Similar design principals are also applied to other agents in the intersection.

\subsection{Rule-based payoff}\label{IVB}
In this subsection, we design the rule-based payoff for the ego, i.e., $J_1^r(a_1,a_2,...a_N)$. Consider an all-way stop sign intersection, then the soft rule to be considered is first-come-first-go. To capture this rule, we define the relative arriving order $r_{ik}$ as our decision variable.

\begin{align}\label{rij}
        &r_{ik}=\begin{cases}1, &\mbox{vehicle $i$ arrives earlier than vehicle $k$},\vspace{1ex}\\
0, &\mbox{vehicle $k$ arrives earlier than vehicle $i$},\vspace{1ex}\end{cases}
\end{align}

With the relative arriving order $r_{ik}$, we can define the first-come-first-go -based payoff for the ego $J_1^r(a_1,a_2,...a_N)$ as
\begin{equation}\label{rule_1_e}
    J_1^r(Y,a_2,...a_N)=0.5,
\end{equation}
\begin{equation}\label{rule_2_e}
    J_1^r(G,a_2,...a_N)=\prod_{k=2}^mr_{ik}.
\end{equation}

Comparing Equations \eqref{rule_1_e} and \eqref{rule_2_e}, we can see that  the ego is motivated to go only if $r_{ik}=1$ for all $k\in \{2,...,m\}$, i.e., the ego has the highest priority compared to all of its first level neighbors. Similar payoff designs are also applied to other vehicles in the intersection.

Substituting the safety-based and rule-based payoffs into Equation (5), the ego's payoff designs are then completed. The constant parameters in payoffs, i.e., $\beta$, $\theta_1$, $\theta_2$, $\theta_3$, $\theta_4$, can be learned from offline supervised learning algorithms developed in  \cite{my_ford}. 

\section{Simulation Studies}\label{V}
In this section, we conduct  simulation studies to verify the effectiveness of the developed hierarchical game in multi-vehicle intersection-crossing scenarios. 

The vehicles' dynamics is described using the kinematic bicycle model \cite{bicycle}.
\begin{equation}
\begin{split}
    \Dot{x}_i&=v_{i}cos(\phi_i+\beta_i),\\
    \Dot{y}_i&=v_{i}sin(\phi_i+\beta_i),\\
    \Dot{\phi}_{i}&=\frac{v}{l_r}sin(\beta_i),\\
    \Dot{v}_i&=a_{i},\\
    \beta_i&=tan^{-1}\left(\frac{l_r}{l_r+l_f}tan(\delta_i)\right),
    \end{split}
\end{equation}
where $i=1,\cdots,N$ represents the agents, $x_i$ and $y_i$ are the longitudinal and lateral
position of the center of mass of  vehicle $i$ along $x$ and $y$ axes, respectively, $v_{i}$ and  $a_{i}$  are  the velocity and acceleration of the center of mass of  vehicle $i$, respectively, $\beta$ is  the angle of the velocity with respect to the longitudinal axis of the vehicle. $\phi_i$ is the inertial
heading. $l_f$ and $l_r$ are the lengths from the center of mass to the front and rear ends of the car, respectively, and are selected as $l_f = l_r = 1.5m$. The inputs of the system are the acceleration $a_i$ and the front steering
angle $\delta_i$. 

In the following simulation studies, we first simulate the cases where three vehicle get involved, to test the functionality of the proposed game. Then we simulate the cases where many vehicles  ($6-10$) are involved, to test the ego's performance in dense traffic scenarios. 

\subsection{Three-vehicle scenarios}
In three-vehicle scenarios, we conduct the following three studies: 1) the test of the functionality of each payoff component; 2) the test of the interaction between the ego and its second level neighbor; and 3) the comparison of the ego's performances in the  hierarchical game and pairwise games.

\textbf{Study 1: Functionality of each payoff component.}
In this study, we simulate the ego's performance when 1) only the rule-based payoff is included, and 2) both safety and rule -based payoffs are included. The rule-based payoff captures the first-come-first-go rule. Figures \ref{rule_1} and  \ref{rule_2} show, respectively, the ego's decision if it comes to the intersection earlier, and later, respectively, than vehicle 2. The red arrow represents the decision of "go" and the red circle represents "yield". As we can see from the figures, the ego's decision is safe and efficient if the ego's first level neighbor (i.e., vehicle 2) obeys the traffic rule. We also simulate the case where vehicle 2 does not obey the traffic rule, as shown in Figure \ref{rule_3}. In this case,  the ego comes first but vehicle $2$ does not yield to the ego. As a result, the ego crashes into the vehicle 2. 
The above simulations show that the rule-based payoff alone is not sufficient for the ego to always generate safe decisions. 

Then we complete the payoff design with both rule-based and safety-based components, and test the scenario shown in Figure \ref{rule_3} again. In this case, due to the safety-based payoff, the ego is able to successfully stop before it crashes into the vehicle 2, even though the vehicle 2 violates the traffic rule. As we can see from this simulation, with both safety-based and rule-based payoff components, the ego is able to not only obey traffic rules but also generate safe decisions subject to other agents' violation of rules. 

\begin{figure}[thpb]
\centering
\subfigure[]{\label{rule_1}
\includegraphics[width=0.23\textwidth]{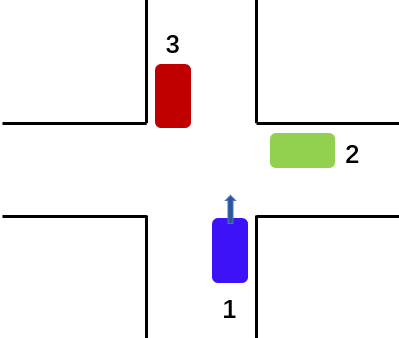}}
\subfigure[]{\label{rule_2}
\includegraphics[width=0.23\textwidth]{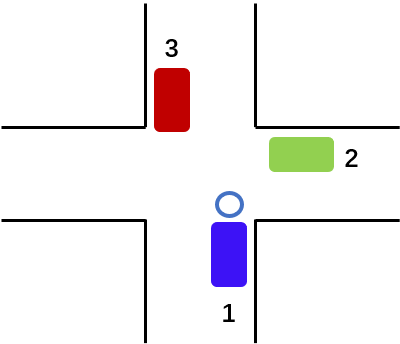}}
\caption{With rule-based payoff only,
 the ego decides to go  when (a) it comes earlier than vehicle 2, and (b) it comes later than vehicle 2.}\label{both_12}
\end{figure}
\begin{figure}[thpb]
\centering
\subfigure[]{\label{rule_3}
\includegraphics[width=0.23\textwidth]{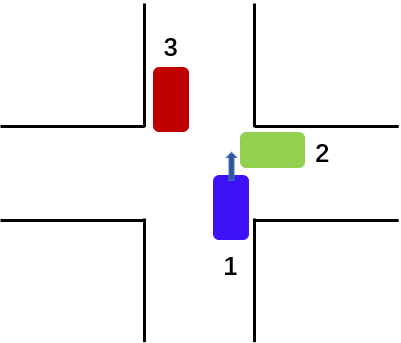}}
\subfigure[]{\label{both_2}
\includegraphics[width=0.23\textwidth]{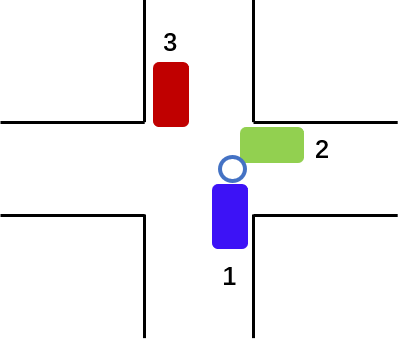}}
\caption{
(a) The ego crashes into vehicle 2 when vehicle 2 violates the traffic rule and the ego's payoff contains the rule-based component only. (b) The ego successfully stops before crashing into the vehicle 2, with both safety and rule -based components. }\label{both_34}
\end{figure}

 \textbf{Study 2: Interaction between the ego and its second level neighbors.}
  In this study, we test how the   second level neighbors affect the ego's decision. Note that the ego does not have trajectory conflicts with its second level neighbors. The results are shown in Figure \ref{study2}. In Figure \ref{study2_1}, the vehicle 3 (i.e., ego's second level neighbor) yields to vehicle 2, and the ego decides to yield as it arrives later than vehicle 2. In Figure \ref{study2_2}, the vehicle 3 decides to go as it comes earlier than vehicle 2, which also motivates the ego to go because it knows that vehicle 2 should yield to vehicle 3 and thus makes itself safe to go.  This simulation shows that  the second level neighbors, although do not have direct trajectory conflicts with the ego, affects the ego's decision. To be more specific, by taking into account the second level neighbors,  the ego's decision is expected to be more efficient  with less unnecessary waiting.   
\begin{figure}[thpb]
\centering
\subfigure[]{\label{study2_1}
\includegraphics[width=0.23\textwidth]{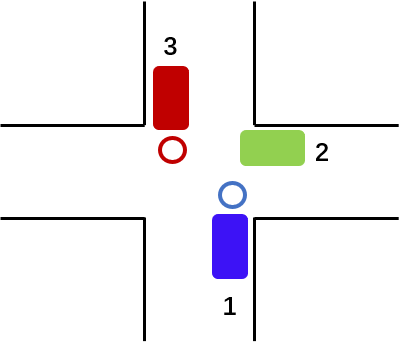}}
\subfigure[]{\label{study2_2}
\includegraphics[width=0.23\textwidth]{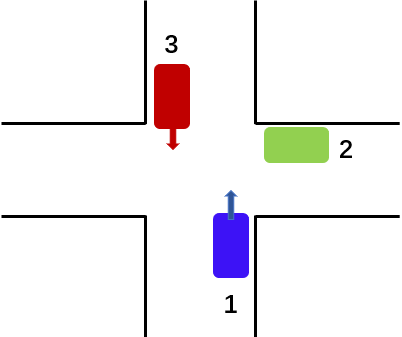}}
\caption{The ego's decisions in Study 2. 
(a) The ego decides to yield when vehicle 3 yields and  vehicle 2 comes earlier. (b) The ego decides to go when vehicle $3$ decides to go. }\label{study2}
\end{figure}

 \textbf{Study 3: Comparison between the hierarchical game and the pairwise game.}
  In this study, we compare the ego's performance when its decisions are generated by the hierarchical game and the pairwise games, respectively. We use the same scenario as shown in  Figure \ref{study2_2}. The hierarchical game generates  the decision of "go" for the ego, as shown in  Figure \ref{study2_2}. However, the pairwise games generate the decision of "yield", as shown in Figure \ref{compare}.  In pairwise games, the ego plays a $2$-player  game with vehicle 2 and 3, respectively, and selects the most conservative action from the two game outcomes. The game with vehicle 2 generates a "yield" decision for the ego, as the ego comes later than vehicle 2, and as such, the ego selects "yield" as its final decision.  From this simulation we can see that the pairwise games only capture local interactions between a pair of agents, and can lead to more conservative decisions compared to the hierarchical game.      

\begin{figure}[thpb]
\centering
\subfigure[]{\label{compare}
\includegraphics[width=0.3\textwidth]{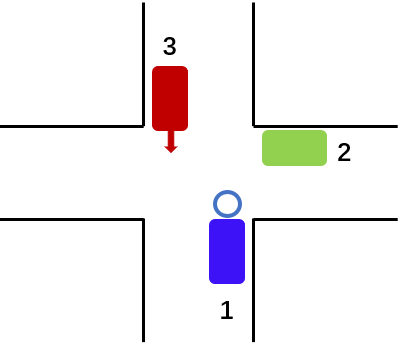}}
\caption{The ego's decision with pairwise games. }
\end{figure}

\subsection{Many-vehicle scenarios}
In this subsection, we test the ego's performance when many vehicles ($N\geq6$) are involved. We conduct the following two studies: 1) the test of the ego's response to surrounding agents' abrupt change of behaviors; and 2) the test of the ego' performances in various  scenarios.  

 \textbf{Study 1: Ego's response to surrounding agents' abrupt change of behaviors.} In this study, we simulate 6-vehicle scenarios, as shown in Figure \ref{study4}. The red vehicle represents the ego, the green ones are the ego's first level neighbors before agent clustering, and the blue one is the ego's second level of neighbor. The white arrows  represent the traffic directions. In the scenario shown in Figure \ref{study4_1}, all surrounding agents move smoothly, and the ego decides to cross the intersection in front of the rightmost  vehicle. In scenario shown in Figure \ref{study4_2},  the  rightmost vehicle speeds up suddenly before the ego crosses it. In this case, the ego is able to successfully stop and yield to the green vehicle as shown in  \ref{study4_2}.   
\begin{figure}[thpb]
\centering
\subfigure[]{\label{study4_1}
\includegraphics[width=0.23\textwidth]{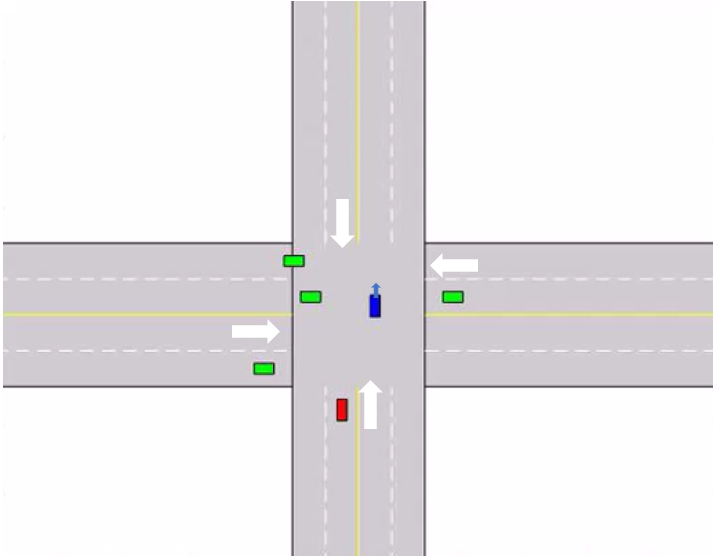}}
\subfigure[]{\label{study4_2}
\includegraphics[width=0.23\textwidth]{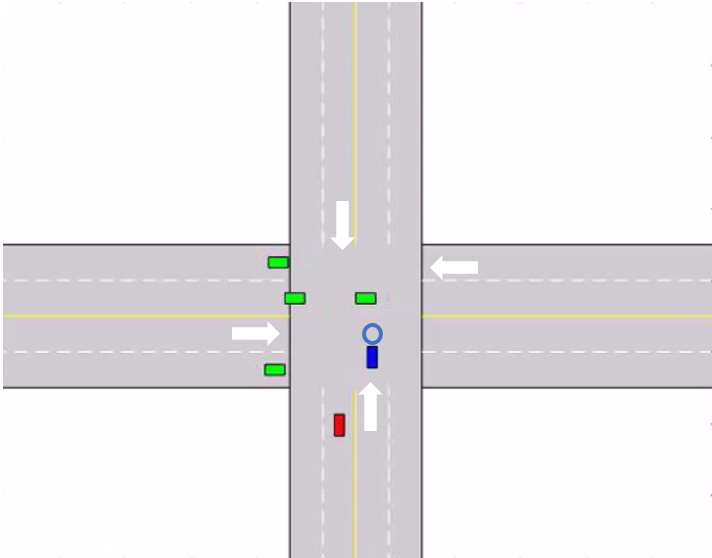}}
\caption{The ego's decision in Study 1. 
(a) The ego crosses the intersection in front of the rightmost vehicle. (b) The ego  yields to respond to the rightmost vehicle's abrupt change of behavior. }\label{study4}
\end{figure}

 \textbf{Study 2: Ego's performance in various complex scenarios.}
 In this study, we test the ego's performance in various intersection-crossing scenarios, when 9-10 vehicles are involved (as shown in Figure \ref{study5}).  In all these tested scenarios, the ego is able to generate both safe and efficient decisions within a short time period ($<10$ milliseconds), which verifies the effectiveness of the developed hierarchical game algorithm.
\begin{figure}[thpb]
\centering
\subfigure[]{\label{study5_1}
\includegraphics[width=0.23\textwidth]{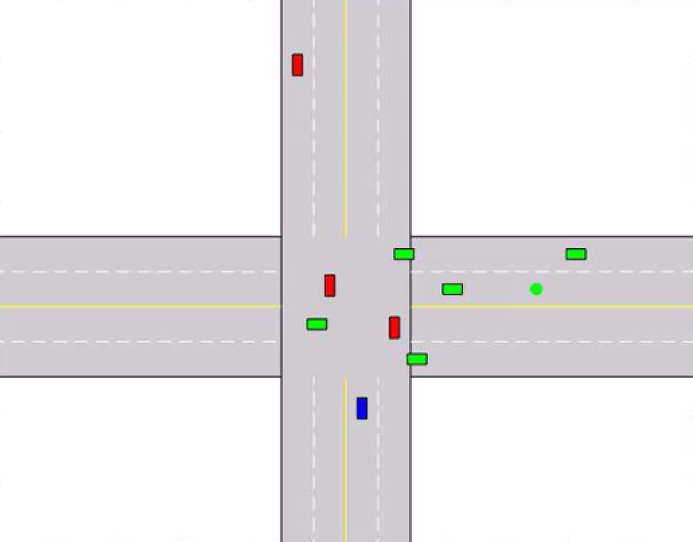}}
\subfigure[]{\label{study5_2}
\includegraphics[width=0.23\textwidth]{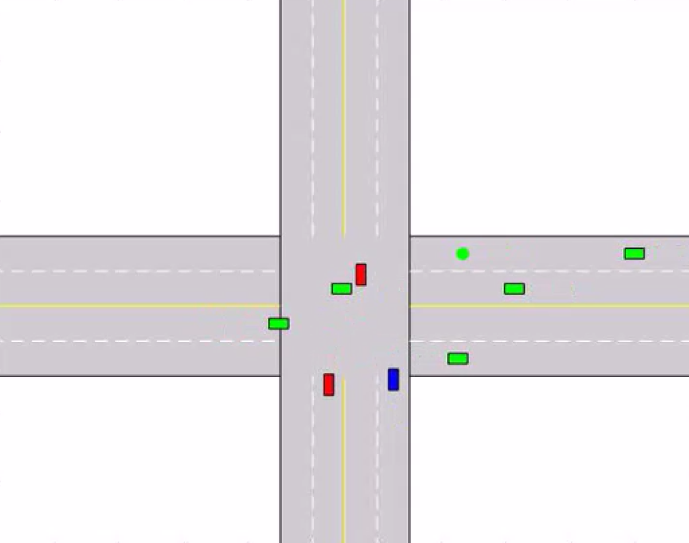}}
\caption{The ego performs well in various scenarios.}\label{study5}
\end{figure}

\section{Conclusions}\label{VI}
This paper proposes a novel hierarchical game to solve multi-agent game-theoretic decision-making in autonomous driving. This hierarchical game features an interaction graph, which characterizes the interaction relationships between the ego and its surrounding agents.  With this hierarchical game, the ego is able to smartly select a limited number of agents as its game players, which significantly reduces the computational complexity compared to the multi-player game where all surrounding agents are included in the set of game players. Compared to the pairwise game, the hierarchical game captures more interaction information, and as such, lead to   more efficient decisions for the ego while guarantees safety. In the future, we will incorporate this hierarchical game framework with trajectory prediction techniques to capture the driving intentions of surrounding agents and facilitate the ego's decision making under uncertain intentions.

\bibliography{main}{}
\bibliographystyle{IEEEtran}
\begin{IEEEbiography}[{\includegraphics[width=1in,height=1.25in,clip,keepaspectratio]{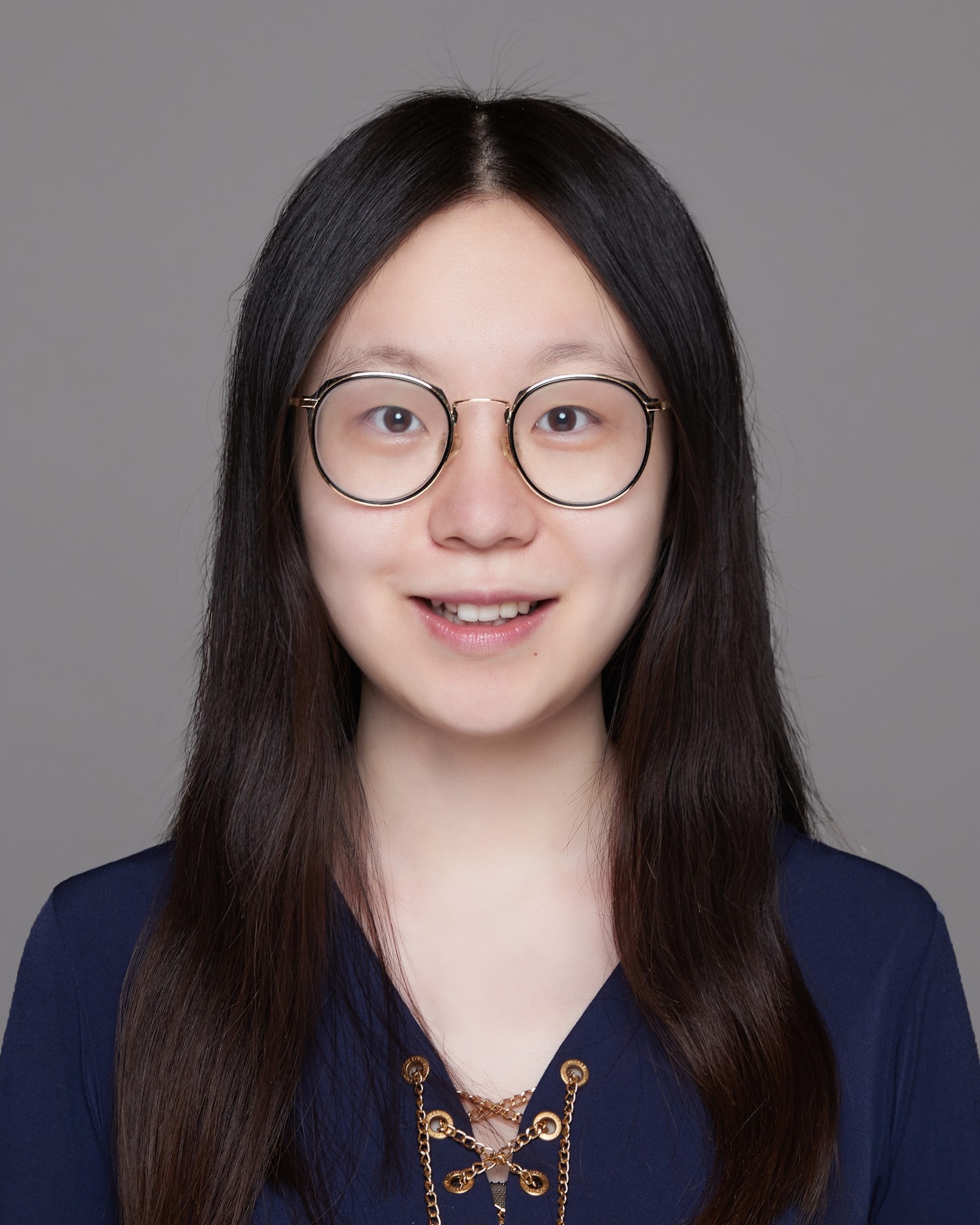}}]{Mushuang Liu} is currently a Postdoctoral Research Fellow in the Aerospace Engineering at the University of Michigan, Ann Arbor, Michigan. She was an Adjunct Professor and Postdoctoral Research Associate in  Electrical Engineering at the University of Texas at Arlington (UTA). She received her B.S. degree in Electrical Engineering from University of Electronic Science and Technology of China, Chengdu, China in 2016, and her Ph.D degree in Electrical Engineering from UTA, TX, USA, in 2020. Her research interests include decision-making in multi-agent systems, optimal control, distributed control, multi-player games, reinforcement learning, and their applications to UAV traffic management, UAV networking, and autonomous driving. Her research has led to over 20 publications and multiple awards, including Lockheed Martin Missiles and Fire Control Graduate Fellowship,  Research in Motion Scholarships, and Office of Graduate Studies Travel Grant Award.
\end{IEEEbiography}

\begin{IEEEbiography}[{\includegraphics[width=1in,height=1.25in,clip,keepaspectratio]{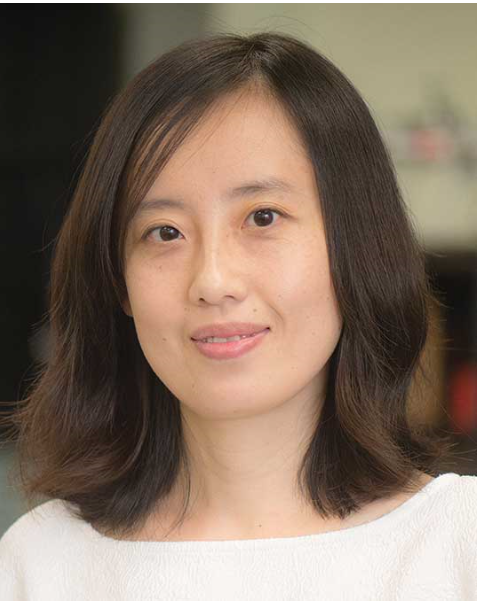}}]{Yan Wan} is currently a Professor in the Electrical Engineering Department at the University of Texas at Arlington. She received her Ph.D. degree in Electrical Engineering from Washington State University in 2008 and then did postdoctoral training at the University of California, Santa Barbara. Her research interests lie in the modeling, evaluation, and control of large-scale dynamical networks, cyber-physical systems, stochastic networks, decentralized control, learning control, networking, uncertainty analysis, algebraic graph theory, and their applications to urban aerial mobility, autonomous driving, robot networking, and air traffic management. Her research has led to over 190 publications and successful technology transfer outcomes. She has received prestigious awards, including the NSF CAREER Award, RTCA William E. Jackson Award, U.S. Ignite and GENI demonstration awards, IEEE WCNC and ICCA Best Paper Award, and Tech Titan of the Future-University Level Award.
\end{IEEEbiography}

\begin{IEEEbiography}[{\includegraphics[width=1in,height=1.25in,clip,keepaspectratio]{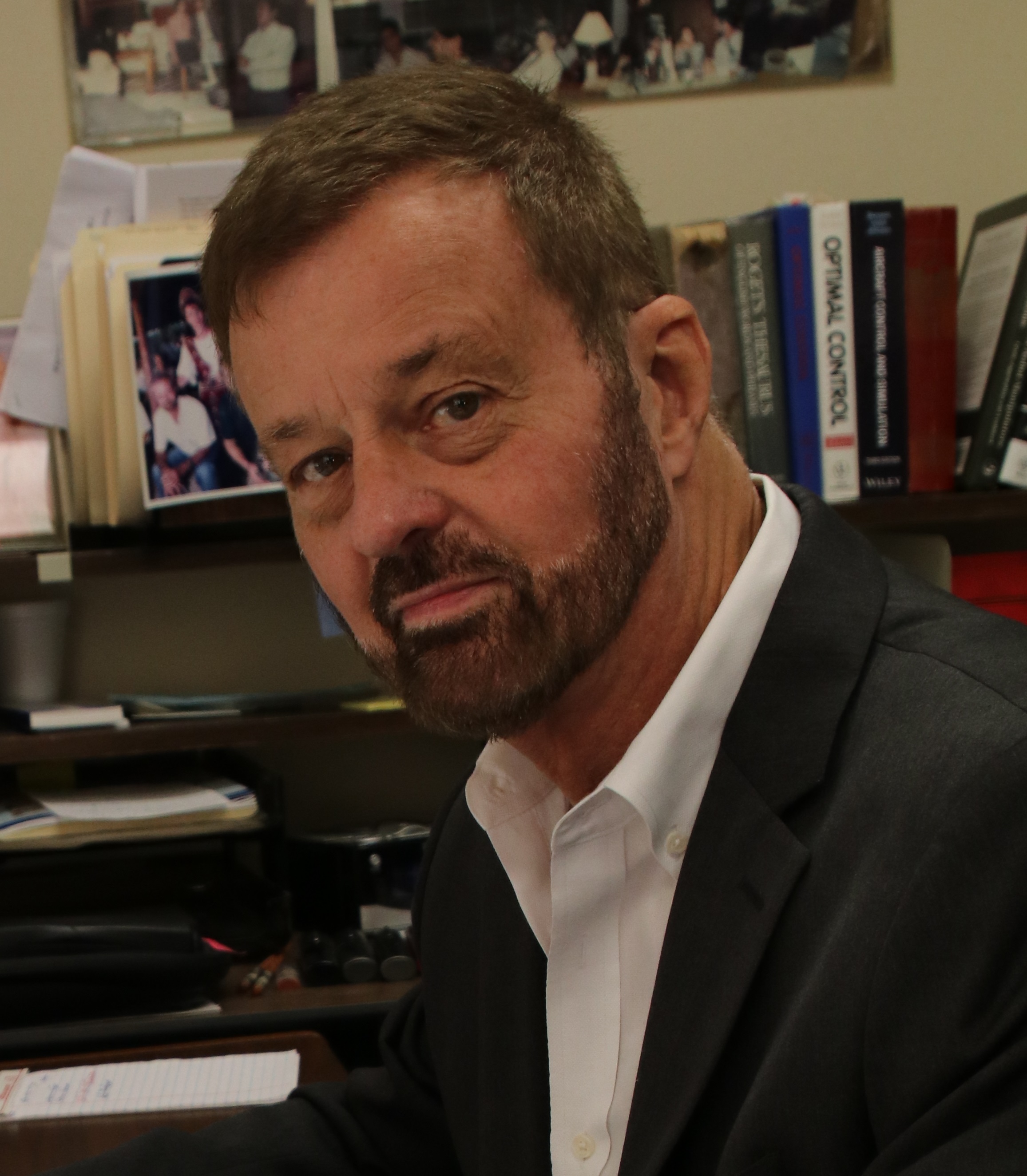}}]{Frank L. Lewis} obtained the Bachelor's Degree in Physics/EE and the MSEE at Rice University, the MS in Aeronautical Engineering from Univ. W. Florida, and the Ph.D. at Ga. Tech.  Fellow, National Academy of Inventors. Fellow IEEE, Fellow IFAC, Fellow AAAS, Fellow U.K. Institute of Measurement \& Control, PE Texas, U.K. Chartered Engineer. UTA Charter Distinguished Scholar Professor, UTA Distinguished Teaching Professor, and Moncrief-O’Donnell Chair at the University of Texas at Arlington Research Institute.  Ranked at position 88 worldwide, 66 in the USA, and 3 in Texas of all scientists in Computer Science and Electronics, by Guide2Research.com. 65,000 Google Scholar Citations. He works in feedback control, intelligent systems, reinforcement learning, cooperative control systems, and nonlinear systems.  He is author of 7 U.S. patents, numerous journal special issues, 445 journal papers, 20 books, including the textbooks Optimal Control, Aircraft Control, Optimal Estimation, and Robot Manipulator Control. He received the Fulbright Research Award, NSF Research Initiation Grant, ASEE Terman Award, Int. Neural Network Soc. Gabor Award, U.K. Inst Measurement \& Control Honeywell Field Engineering Medal, IEEE Computational Intelligence Society Neural Networks Pioneer Award, AIAA Intelligent Systems Award, AACC Ragazzini Award.  He has received over $\$$12M in 100 research grants from NSF, ARO, ONR, AFOSR, DARPA, and USA industry contracts.  Helped win the US SBA Tibbets Award in 1996 as Director of the UTA Research Institute SBIR Program.  
\end{IEEEbiography}
\begin{IEEEbiography}[{\includegraphics[width=1in,height=1.25in,clip,keepaspectratio]{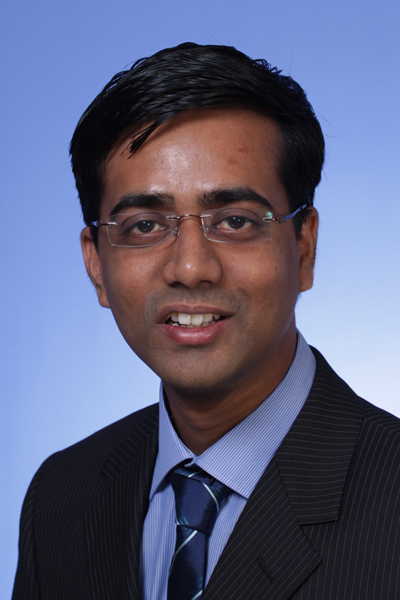}}]{Subramanya Nageshrao}  received his Ph.D. degree from the Delft Center for Systems and Control, Delft University of Technology, Delft, The Netherlands, in 2016. He is currently a Research Engineer at Ford Green Field Labs, Palo Alto, California. Prior to that he worked as a post-doctoral research scholar at University of Michigan Ann-Arbor. He is actively leveraging machine learning methods particularly reinforcement learning techniques to build next generation state-of-the-art advanced driver assistance technologies.
\end{IEEEbiography}
\begin{IEEEbiography}[{\includegraphics[width=1in,height=1.25in,clip,keepaspectratio]{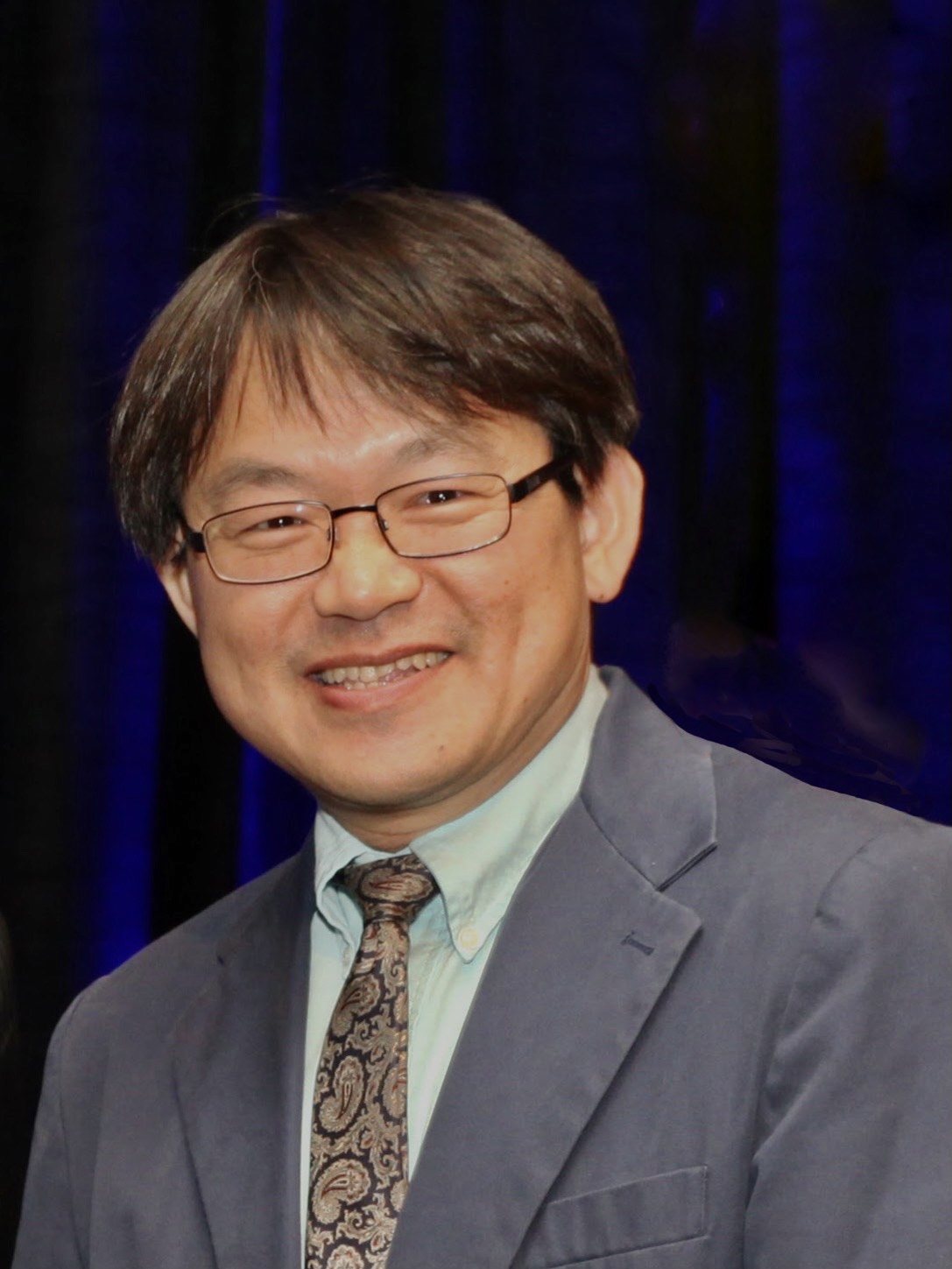}}]{Hongtei Eric Tseng} received his B.S. degree from National Taiwan University, Taipei, Taiwan in 1986.  He received his M.S. and Ph.D. degrees from the University of California, Berkeley in 1991 and 1994, respectively, all in Mechanical Engineering.  In 1994, he joined Ford Motor Company. 

At Ford, he is currently a Senior Technical Leader of Controls and Automated Systems in Research and Advanced Engineering.  Many of his contributed technologies led to production vehicles implementation. His technical achievements have been recognized internally seven times with Ford’s highest technical award - the Henry Ford Technology Award, as well as externally by the American Automatic Control Council with Control Engineering Practice Award in 2013.  Eric has over 100 US patents and over 120 publications. He is an NAE Member.
\end{IEEEbiography}
\begin{IEEEbiography}[{\includegraphics[width=1in,height=1.5in,clip,keepaspectratio]{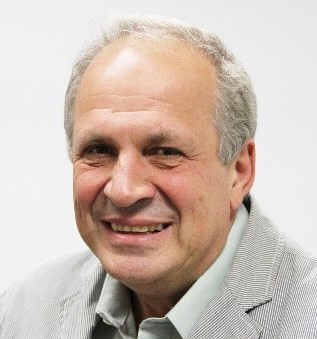}}]{Dimitar Filev} is Senior Henry Ford Technical Fellow in Control and AI with Research $\&$ Advanced Engineering – Ford Motor Company. His research is in computational intelligence, AI and intelligent control, and their applications to autonomous driving, vehicle systems, and automotive engineering.  He holds over 100 granted US patents and has been awarded with the IEEE SMCS 2008 Norbert Wiener Award and the 2015 Computational Intelligence Pioneer’s Award. Dr. Filev is a Fellow of the IEEE and a member of the National Academy of Engineering. He was President of the IEEE Systems, Man, \& Cybernetics Society (2016-2017).
\end{IEEEbiography}
\end{document}